\documentclass[copyright,creativecommons]{eptcs}

\providecommand{\event}{WORDS 2011}

\usepackage[latin1]{inputenc}
\usepackage[T1]{fontenc}

\usepackage{xspace}
\usepackage{longtable}
\usepackage{array}
\usepackage{comment}

\usepackage{amsmath}
\usepackage{amssymb}
\usepackage{amsthm}

\usepackage{paralist}

\input{define}

\draftoff

\title{Bounded Parikh Automata}
\def\titlerunning{Bounded Parikh Automata}

\author{Micha\"el Cadilhac$^1$
  \and
  Alain Finkel$^2$
  \and
  Pierre McKenzie$^1$
  \switchtoinstitutes
  $^1$~~DIRO,
  Universit\'e de Montr\'eal
  \email{\texttt{\{cadilhac, mckenzie\}@iro.umontreal.ca}}
  \and
  $^2$~~LSV, ENS Cachan \& CNRS
  \email{finkel@lsv.ens-cachan.fr}}
\def\authorrunning{M. Cadilhac, A. Finkel, and P. McKenzie}

\def\revisionset! #1 !{\gdef\therevision{#1}}
\def\revision$LastChangedRevision:#1${\revisionset!#1!}
\revision$LastChangedRevision: 452 $

\begin{document}
\maketitle

\begin{abstract}
  The Parikh finite word automaton model (PA) was introduced and studied by
  \KaR~\cite{klaedtke-ruess03}.  Here, by means of related models, it is
  shown that the bounded languages recognized by PA are the same as those
  recognized by deterministic PA.  Moreover, this class of languages is the
  class of bounded languages whose set of iterations is semilinear.
\end{abstract}


\section{Introduction}



\paragraph{Motivation.}  Adding features to finite automata in order to
capture situations beyond regularity has been fruitful to many areas of
research.  Such features include making the state sets infinite, adding power
to the logical characterizations, having the automata operate on infinite
domains rather than finite alphabets, adding stack-like mechanisms,
etc. (See, e.g.,~\cite{fischer65,bojanczyk09,kaminski94,alur-madhusudan09}.)
Model checking and complexity theory below NC$^2$ are areas that have
benefited from an approach of this type (e.g.,~\cite{klarlund89,str94}).  In
such areas, \emph{determinism} plays a key role and is usually synonymous
with a clear understanding of the situation at hand, yet often comes at the
expense of other properties, such as expressiveness.  Thus, cases where
determinism can be achieved without sacrificing other properties are of
particular interest.

\paragraph{Context.} \KaR introduced the Parikh automaton (PA) as an
extension of the finite automaton~\cite{klaedtke-ruess03}.  A PA is a pair
$(A, C)$ where $C$ is a semilinear subset of $\bbn^d$ and $A$ is a finite
automaton over $(\Sigma \times D)$ for $\Sigma$ a finite alphabet and $D$ a
finite subset of $\bbn^d$. The PA accepts the word $w_1\cdots w_n \in
\Sigma^*$ if $A$ accepts a word $(w_1, \bar{v_1})\cdots(w_n, \bar{v_n})$ such
that $\sum \bar{v_i} \in C$.  \KaR used the PA to characterize an extension
of (existential) monadic second-order logic in which the cardinality of sets
expressed by second-order variables is available.  To use PA as symbolic
representations for model checking, the closure under the Boolean operations
is needed; unfortunately, PA are not closed under complement.  Moreover,
although they allow for great expressiveness, they are not determinizable.

\paragraph{Bounded and semilinear languages.}  Bounded languages were
defined by Ginsburg and Spanier in 1964~\cite{ginsburg-spanier64} and
intensively studied in the sixties.  Recently, they played a new role in the
theory of acceleration in regular model
checking~\cite{finkel-iyer-sutre03,bouajjani-habermehl98}.
A language $L \subseteq \Sigma^*$ is \emph{bounded} if there exist words
$w_1,w_2,\ldots,w_n \in \Sigma^*$ such that $L \subseteq w_1^*w_2^*\cdots
w_n^*$.  Bounded context-free languages received much attention thanks to
better decidability properties than context-free languages~\cite{ginsburg66}
(e.g., inclusion between two context-free languages is decidable if one of
them is bounded, while it is undecidable in the general case).
%
%
Moreover, given a context-free language it is possible to decide whether it
is bounded~\cite{ginsburg-spanier64}.
Context-free languages have a 
semilinear Parikh image ~\cite{parikh66}.
Connecting semilinearity and boundedness,
the class $\BSL$ of bounded languages $L \subseteq w_1^*\cdots w_n^*$, for
which $\{(i_1, \ldots, i_n) \st w_1^{i_1}\cdots w_n^{i_n} \in L\}$ is a
semilinear set, has been studied intensively
(e.g.,~\cite{ginsburg-spanier64,ginsburg66,ibarra-su-dang-bultan-kemmerer02,dalessandro-varricchio08}).

\paragraph{Our contribution.}  We study PA whose language is bounded. Our main result is that bounded PA languages are also accepted by deterministic PA, and that they correspond
exactly to BSL.  Moreover, we give a precise deterministic PA form into which
every such language can be cast, thus relating our findings to another model
in the literature, CQDD~\cite{bouajjani-habermehl98}.  To the best of our
knowledge, this is the first characterization of BSL by means of a
deterministic model of one-way automata.



\section{Preliminaries}\label{sec:prelim}

We write $\bbn$ for the nonnegative integers and $\bbn^+$ for $\bbn \setminus
\{0\}$.  Let $d > 0$ be an integer.  Vectors in $\bbn^d$ are noted with a bar
on top, e.g., $\bar{v}$ whose elements are $v_1, \ldots, v_d$.  We write
$\bar{e_i} \in \{0, 1\}^d$ for the vector having a $1$ only in position $i$.
We view $\bbn^d$ as the additive monoid $(\bbn^d, +)$.  For a monoid $(M,
\cdot)$ and $S \subseteq M$, we write $S^*$ for the monoid generated by $S$,
i.e., the smallest submonoid of $(M, \cdot)$ containing~$S$.  A subset $C$ of
$\bbn^d$ is \emph{linear} if there exist $\bar{c} \in \bbn^d$ and a finite $P
\subseteq \bbn^d$ such that $C = \bar{c} + P^*$.  The subset $C$ is said to
be \emph{semilinear} if it is a finite union of linear sets.  Semilinear sets
are the sets expressible by a quantifier-free first-order formula $\phi$
which uses the function symbol $+$, the congruence relations $\equiv_i$, for
$i \geq 2$, and the order relation $<$ (see, e.g.,~\cite{enderton72}).  More
precisely, a subset $C$ of $\bbn^d$ is semilinear iff there is such a formula
with $d$ free variables, with $(x_1, \ldots, x_d) \in C \Leftrightarrow \bbn
\models \phi(x_1, \ldots, x_d)$, where $\bbn$ is the standard model of
arithmetic.

Let $\Sigma = \{a_1, \ldots, a_n\}$ be an (ordered) alphabet, and write
$\eps$ for the empty word.  The \emph{Parikh image} is the morphism
$\pkh\colon \Sigma^* \to \bbn^n$ defined by $\pkh(a_i) = \bar{e_i}$, for $1
\leq i \leq n$.  A language $L \subseteq \Sigma^*$ is said to be
\emph{semilinear} if $\pkh(L) = \{\pkh(w) \st w \in L\}$ is semilinear.

A language $L \subseteq \Sigma^*$ is \emph{bounded}~\cite{ginsburg-spanier64}
if there exist $n > 0$ and a sequence of words $w_1, \ldots, w_n \in
\Sigma^+$, which we call \emph{a socle of $L$}, such that $L \subseteq
w_1^*\cdots w_n^*$.  The \emph{iteration set} of $L$ w.r.t.\ this socle is
defined as $\iter_{(w_1, \ldots, w_n)}(L) = \{(i_1, \ldots, i_n) \in \bbn^n
\st w_1^{i_1}\cdots w_n^{i_n} \in L\}$.  $\BDD$ stands for the class of
bounded languages.  We denote by $\BSL$ the class of \emph{bounded semilinear
  languages}, defined as the class of languages $L$ for which there exists a
socle $w_1, \ldots, w_n$ such that $\iter_{(w_1, \ldots, w_n)}(L)$ is
semilinear; in particular, the Parikh image of a language in $\BSL$ is
semilinear.

We then fix our notation about automata.  An automaton is a quintuple $A =
(Q, \Sigma, \delta, q_0, F)$ where $Q$ is the finite set of states, $\Sigma$
is an alphabet, $\delta \subseteq Q \times \Sigma \times Q$ is the set of
transitions, $q_0 \in Q$ is the initial state and $F \subseteq Q$ are the
final states.  For a transition $t \in \delta$, where $t = (q, a, q')$, we
define $\orig(t) = q$ and $\dest(t) = q'$.  Moreover, we define $\mu_A\colon
\delta^* \to \Sigma^*$ to be the morphism given by $\mu_A(t) = a$, and we
write $\mu$ when $A$ is clear from the context.  A \emph{path} on $A$ is a
word $\pi = t_1\cdots t_n \in \delta^*$ such that $\dest(t_i) =
\orig(t_{i+1})$ for $1 \leq i < n$; we extend $\orig$ and $\dest$ to paths,
letting $\orig(\pi) = \orig(t_1)$ and $\dest(\pi) = \dest(t_n)$.  We say that
$\mu(\pi)$ is the \emph{label} of $\pi$.  A path $\pi$ is said to be
\emph{accepting} if $\orig(\pi) = q_0$ and $\dest(\pi) \in F$; we write
$\apaths(A)$ for the language over $\delta$ of accepting paths on $A$.  We
write $L(A)$ for the language of $A$, i.e., the labels of the accepting
paths.  The automaton $A$ is said to be \emph{deterministic} if $(p, a, q),
(p, a, q') \in \delta$ implies $q = q'$.  An $\eps$-automaton is an automaton
$A = (Q, \Sigma, \delta, q_0, F)$ as above, except with $\delta \subseteq
Q\times(\Sigma \cup \{\eps\})\times Q$ so that in particular $\mu_A$ becomes
an erasing morphism.

Following~\cite{bouajjani-habermehl98, BFLS05-atva, DFGD-jancl10}, we say
that an automaton is \emph{flat} if it consists of a set of states $q_0,
\ldots, q_n$ such that $q_0$ is initial, $q_n$ is final, there is one and
only one transition between $q_i$ and $q_{i+1}$, and there may exist, at
most, a path of fresh states from $q_j$ to $q_i$ iff $i \leq j$ and no such
path exists with one of its ends between~$i$ and~$j$.  Note that no nested
loop is allowed in a flat automaton, and that the language of a flat
automaton is bounded.  Flat automata are called \emph{restricted simple} in
\cite{bouajjani-habermehl98}.

Let $\Sigma$ and $\Tau$ be two alphabets.  Let $A$ be an automaton over the
alphabet $(\Sigma \cup \{\eps\}) \times (\Tau \cup \{\eps\})$, where the
concatenation is seen as $(u_1, v_1).(u_2, v_2) = (u_1u_2, v_1v_2)$.  Then
$A$ defines the \emph{rational transduction} $\tau_A$ from languages $L$ on
$\Sigma$ to languages on $\Tau$ given by $\tau_A(L) = \{v \in \Tau^* \st
(\exists u \in L)[(u, v) \in L(A)]\}$.  Closure under rational transduction
for a class $\calC$ is the property that for any language $L \in \calC$ and
any automaton $A$, $\tau_A(L) \in \calC$.  We say that $\tau_A$ is a
\emph{deterministic} rational transduction if $A$ is deterministic with
respect to the first component of its labels, i.e., if $(p, (a, b), q)$ and
$(p, (a, b'), q')$ are transitions of $A$, then $b = b'$ and $q = q'$.


\section{Parikh automata and constrained automata}

The following notations will be used in defining Parikh finite word automata
(PA) formally.  Let $\Sigma$ be an alphabet, $d \in \bbn^+$, and $D$ a finite
subset of $\bbn^d$.  Following~\cite{klaedtke-ruess03}, let $\proj\colon
(\Sigma \times D)^* \to \Sigma^*$ and $\pext\colon (\Sigma \times D)^* \to
\bbn^d$ be two morphisms defined, for $\ell = (a, \bar{v}) \in \Sigma \times
D$, by $\proj(\ell) = a$ and $\pext(\ell) = \bar{v}$.  The function $\proj$
is called the \emph{projection on $\Sigma$} and the function $\pext$ is
called the \emph{extended Parikh image}.
As an example, for a word $\omega \in \{(a_i, \bar{e_i}) \st 1 \leq i \leq
n\}^*$, the value of $\pext(\omega)$ is the Parikh image of $\proj(\omega)$.

\begin{definition}[Parikh automaton~\cite{klaedtke-ruess03}]
  Let $\Sigma$ be an alphabet, $d \in \bbn^+$, and $D$ a finite subset of
  $\bbn^d$.  A \emph{Parikh automaton (PA)} of dimension $d$ over $\Sigma
  \times D$ is a pair $(A, C)$ where $A$ is a finite automaton over $\Sigma
  \times D$, and $C \subseteq \bbn^d$ is a semilinear set.  The PA language
  is $L(A, C) = \{\proj(\omega) \st \omega \in L(A) \land \pext(\omega) \in
  C\}$.%
  
  The PA is said to be \emph{deterministic (DetPA)} if for every state $q$ of
  $A$ and every $a \in \Sigma$, there exists at most one pair $(q', \bar{v})$
  with $q'$ a state and $\bar{v} \in D$ such that $(q, (a, \bar{v}), q')$ is
  a transition of $A$.
The PA is said to be \emph{flat} if $A$ is flat.
We write $\LPA$ (resp. $\LDPA$) for the class of
  languages recognized by PA (resp. DetPA).
\end{definition}

In~\cite{cadilhac-finkel-mckenzie11}, PA are characterized by the following
simpler model:
\begin{definition}[Constrained automaton~\cite{cadilhac-finkel-mckenzie11}]
  A \emph{constrained automaton (CA)} over an alphabet~$\Sigma$ is a pair
  $(A, C)$ where $A$ is a finite automaton over $\Sigma$ with $d$
  transitions, and $C \subseteq \bbn^d$ is a semilinear set.  Its language is
  $L(A, C) = \{\mu(\pi) \st \pi \in \apaths(A) \land \pkh(\pi) \in C\}$.
  
  The CA is said to be \emph{deterministic (DetCA)} if $A$ is deterministic.
  An $\eps$-CA is defined as a CA except that $A$ is an $\eps$-automaton.
  Finally, the CA is said to be \emph{flat} if $A$ is flat.
\end{definition}

\begin{theorem}[\cite{cadilhac-finkel-mckenzie11}]\label{thm:altdef}%
  CA and $\eps$-CA define the same class of languages, and the following are
  equivalent for any $L \subseteq \Sigma^*$:
  \begin{compactenum}[(i)]
  \item $L \in \LPA$ (resp.\ $\in \LDPA$);
  \item $L$ is the language of an $\eps$-CA (resp. deterministic CA).
  \end{compactenum}
\end{theorem}


\section{Bounded Parikh automata}

Let $\LBPA$ be the set $\LPA \cap \BDD$ of bounded PA languages, and
similarly let $\LBDPA$ be $\LDPA \cap \BDD$.

Theorem~\ref{thm:lpba} below characterizes $\LBPA$ as the class $\BSL$ of
bounded semilinear languages.  In one direction of the proof, given $L \in
\BSL$, an $\eps$-CA $(A,C)$ for $L$ is constructed.  We describe this simple
construction here and, for future reference, we will call $(A, C)$ the
\emph{canonical $\eps$-CA} of $L$.

Let $L \in \BSL$, there exists a socle $w_1, \ldots, w_n$ of $L$ (minimum
under some ordering for uniqueness) such that $E = \iter_{(w_1, \ldots,
  w_n)}(L)$ is a semilinear set.  Informally, the automaton $A$ will consist
of $n$ disjoint cycles labeled $w_1, \ldots, w_n$ and traversed at their
origins by a single $\eps$-labeled path leading to a unique final state.
Then $C$ will be defined to monitor $(\#t_1,\ldots,\#t_n)$ in accordance with
$E$, where $t_i$ is the first transition of the cycle for $w_i$ and $\#t_i$
is the number of occurrences of $t_i$ in a run of $A$.  Formally, let $k_j =
\sum_{1\leq i\leq j} |w_i|$, $0\leq j\leq n$, and set $Q = \{1, \ldots,
k_n\}$. Then $A$ is the $\eps$-automaton $(Q, \Sigma, \delta, q_0, F)$ where
$q_0 = 1$, $F = \{k_{n-1}+1\}$ and for any $1 \leq i< n$, there is a
transition $(k_{i-1}+1,\eps,k_i+1)$ and for any $1 \leq i\leq n$ a cycle
$t_i\rho_i$ labeled $w_i$ on state $k_{i-1}+1$, where $t_i$ is a transition
and $\rho_i$ a path.  Then $C \subseteq \bbn^{|\delta|}$ is the semilinear
set defined by $(\#t_1,\ldots,\#t_2,\ldots,\ldots,\#t_n,\ldots) \in C$ iff
$(\#t_1, \#t_2,\ldots,\#t_n) \in E$.

\begin{theorem}\label{thm:lpba}
  $\LBPA = \BSL$.
\end{theorem}
\begin{proof}
  ($\LBPA \subseteq \BSL$): Let $L \subseteq \Sigma^*$ be a bounded language
  of $\LPA$, and $w_1, \ldots, w_n$ be a socle of $L$.  Define $E =
  \iter_{(w_1, \ldots, w_n)}(L)$.  Let $T = \{a_1, \ldots, a_n\}$ be a fresh
  alphabet, and let $h\colon T^* \to \Sigma^*$ be the morphism defined by
  $h(a_i) = w_i$.  Then the language $L' = h^{-1}(L) \cap (a_1^*\cdots
  a_n^*)$ is in $\LPA$ by closure of $\LPA$ under inverse morphisms and
  intersection~\cite{klaedtke-ruess03}.  But $\pkh(L') = E$, and as any
  language of $\LPA$ is semilinear~\cite{klaedtke-ruess03}, $E$ is
  semilinear.  Thus the iteration set $E$ of the bounded language $L$ with
  respect to its socle $w_1, \ldots, w_n$ is semilinear, and this is the
  meaning of $L$ belonging to $\BSL$.

  ($\BSL\subseteq \LBPA$): Let $L \in \BSL$.  Of course $L\in\BDD$.  We leave
  out the simple proof that $L$ equals the language of its ``canonical
  $\eps$-CA'' constructed prior to Theorem~\ref{thm:lpba}.  Since $\eps$-CA
  and PA capture the same languages by Theorem~\ref{thm:altdef}, $L\in\LPA$.
\end{proof}

Theorem~\ref{thm:lpba} and the known closure properties of $\BDD$ and $\LPA$
imply:
\begin{proposition}
  $\BSL$ is closed under union, intersection, concatenation.
\end{proposition}

\section{Bounded Parikh automata are determinizable}

Parikh automata cannot be made deterministic in general.  Indeed,
\KaR~\cite{klaedtke-ruess03} have shown that $\LDPA$ is closed under
complement while $\LPA$ is not, so that $\LDPA \subsetneq \LPA$,
and~\cite{cadilhac-finkel-mckenzie11} further exhibits languages witnessing
the separation.  In this section, we show that PA can be determinized when
their language is bounded:
\begin{theorem}\label{thm:main}
  Let $L \in \LBPA$.  Then $L$ is the union of the languages of flat DetCA
  (i.e., of DetCA with flat underlying automata as defined in
  Section~\ref{sec:prelim}). 
  In particular, as $\LBDPA$ is closed under union, $\LBPA = \LBDPA$.
\end{theorem}
\begin{proof}
  Let $L \in \LBPA$.  By Theorem~\ref{thm:lpba}, $L\in\BSL$.  Recall the
  canonical $\eps$-CA of a $\BSL$ language constructed prior to
  Theorem~\ref{thm:lpba}.  Inspection of the canonical $\eps$-CA $(A,C)$ of
  $L$ reveals a crucial property which we will call
  ``constraint-determinism,'' namely, the property that no two paths $\pi_1$
  and $\pi_2$ in $\apaths(A)$ for which $\mu_A(\pi_1)=\mu_A(\pi_2)$ can be
  distinguished by $C$ (i.e., the property that for any two paths  $\pi_1$
  and $\pi_2$ in $\apaths(A)$ for which $\mu_A(\pi_1)=\mu_A(\pi_2)$,
  $\Phi(\pi_1)\in C$ iff $\Phi(\pi_2)\in C$).  To
  see this property, note that if two accepting paths $\pi_1$ and $\pi_2$ in
  the $\eps$-CA have the same label $w$, then $\pi_1$ and $\pi_2$ describe
  two ways to iterate the words in the socle of $L$.  As the semilinear set
  $C$ describes \emph{all} possible ways to iterate these words to get a
  specific label, $\pkh(\pi_1) \in C$ iff $\pkh(\pi_2) \in C$.

  We will complete the proof in two steps.  First, we will show
  (Lemma~\ref{lem:epsca}) that \emph{any} $\eps$-CA having the
  constraint-determinism property can be simulated by a suitable
  deterministic extension of the PA model, to be defined next.  Second
  (Lemma~\ref{lem:dapafmbd=dpabd}), we will show that any bounded language
  accepted by such a suitable deterministic extension of the PA model
  is a finite union of languages of flat DetCA, thus concluding the proof.
\end{proof}

Our PA extension leading to Lemma~\ref{lem:epsca} hinges on associating 
affine functions (rather than vectors) to PA transitions.
In the following, we consider the vectors in $\bbn^d$ to be \emph{column}
vectors.  Let $d > 0$.  A function $f\colon \bbn^d \to \bbn^d$ is a (total
and positive) \emph{affine function} of dimension~$d$ if there exist a matrix
$M \in \bbn^{d\times d}$ and $\bar{v} \in \bbn^d$ such that for any $\bar{x}
\in \bbn^d$, $f(\bar{x}) = M.\bar{x} + \bar{v}$.  We note $f = (M, \bar{v})$
and write $\func_d$ for the set of such functions; we view $\func_d$ as the
monoid $(\func_d, \diamond)$ with $(f \diamond g) (\bar{x}) = g(f(\bar{x}))$.

Intuitively, an affine Parikh automaton will be defined as a finite automaton
that also operates on a tuple of counters.
Each transition of the automaton will blindly apply an affine transformation
to the tuple of counters.
A word $w$ will be deemed accepted by the affine Parikh automaton iff some
run of the finite automaton accepting $w$ also has the cumulative effect of
transforming the tuple of counters, initially $\bar{0}$, to a tuple
belonging to a prescribed semilinear set.

\begin{definition}[Affine Parikh automaton, introduced
  in~\cite{cadilhac-finkel-mckenzie11}]
  An \emph{affine Parikh automaton (APA)} of dimension $d$ is a triple $(A,
  U, C)$ where $A$ is an automaton with transition set $\delta$, $U$ is a
  morphism from $\delta^*$ to $\func_d$ and $C \subseteq \bbn^d$ is a
  semilinear set; recall that $U$ need only be defined on $\delta$.  The
  language of the APA is $L(A, U, C) = \{\mu(\pi) \st \pi \in \apaths(A)
  \land (U(\pi))(\bar{0}) \in C\}$.

  The APA is said to be \emph{deterministic (DetAPA)} if $A$ is.  We write
  \emph{$\LAPA$} (resp. $\LDAPA$) for the class of languages recognized by
  APA (resp. DetAPA).

  The \emph{monoid of the APA}, written $\mon(U)$, is the multiplicative
  monoid of matrices generated by $\{M \st (\exists t)(\exists \bar{v})[U(t)
  = (M, \bar{v})]\}$.
\end{definition}

\begin{lemma}\label{lem:epsca}%
  An $\eps$-CA $(A, C)$ having the constraint-determinism property (see
  Theorem~\ref{thm:main}) can be simulated (meaning that the language of
  the $\eps$-CA is preserved) by a DetAPA $(A', U, E)$ whose
  monoid is finite and such that $L(A) = L(A')$.
\end{lemma}
\begin{proof}
  We outline the idea before giving the details.  Let $(A,C)$ be the
  $\eps$-CA.  We first apply the standard subset construction and obtain a
  deterministic FA $\und A$ equivalent to $A$.  Consider a state $\und q$ of
  $\und A$. 
  Suppose that after reading some word $w$ leading $\und A$ into state $\und q$
  we had, for each $q\in \und q$, the
  Parikh image $\bar{c_{w,q}}$ (counting transitions in $A$, i.e.,
  recording the occurrences of each transition in $A$) of \emph{some} initial
  $w$-labeled path leading $A$ into state $q$.
Suppose that $(\und q, a, \und r)$ is a transition in $\und A$.
How can we compute, for each $q'\in \und r$, the Parikh image
$\bar{c_{wa,q'}}$ of \emph{some} initial $wa$-labeled path leading $A$ into
$q'$?
It suffices to pick any $q\in \und q$ for which some $a$-labeled path leads
$A$ from $q$ to $q'$ (possibly using the $\eps$-transitions in $A$)
and to add to $\bar{c_{w,q}}$ the contribution of this $a$-labeled
path.
  A DetAPA transition on $a$ is
  well-suited to mimic this computation, since an affine transformation
  can first ``flip'' the current Parikh
  $q$-count tuple ``over'' to the
  Parikh $q'$-count tuple and then add to it the $q$-to-$q'$ contribution.
Hence a DetAPA $(\und A,\cdot,\cdot)$ upon reading a word $w$ leading to its
state $\und q$ is able to 
keep track, for each $q\in \und q$, of the Parikh image of \emph{some} initial
$w$-labeled path leading $A$ into $q$.
We need constraint-determinism only to reach the final conclusion: if a
word $w$ leads $\und A$ into a final state $\und q$, then some $q\in \und q$
is final in $A$, and
because of constraint-determinism, imposing membership in $C$ for the Parikh
image of the particular initial $w$-labeled path leading $A$ to $q$ kept
track of by the DetAPA 
is as good as imposing membership in $C$ for the Parikh image of any other
initial $w$-labeled path leading $A$ to $q$.


  We now give the details.  Say $A=(Q,\Sigma, \delta, q_0, F)$, and identify
  $Q$ with $\{1, \ldots, |Q|\}$.

  For $p, q \in Q$, and $a \in \Sigma$, define $S(p, q, a)$ to be a shortest
  path (lexicographically smallest among shortest paths, for definiteness;
  this path can be longer than one because of $\eps$-transitions)
  from $p$ to $q$ labeled by $a$, or $\bot$ if none exists.  Let $\und A =
  (2^Q, \Sigma, \und\delta, \und{q_0}, \und F)$ be the deterministic version
  of $A$ defined by $\und{q_0} = \{q_0\}$, $\und\delta(\und p, a) = \{q \in Q
  \st (\exists p \in \und p)[S(p,q,a) \neq \bot]\}$ and $\und F = \{\und q
  \st \und q \cap F \neq \emptyset\}$; thus $L(\und A)$ is bounded.  Note
  that, by construction, for any path $\pi$ in $A$ from $q_0$ to $q$, there
  exists a path $\und\pi$ in $\und A$ from $\und{q_0}$ to a state $\und q$
  such that $q \in \und q$ and $\mu_A(\pi) = \mu_{\und A}(\und \pi)$.

  We now attach a function to each transition of $\und A$, where the functions
  are of dimension $(|Q|.|\delta| + 1)$.  We first define $V\colon \und
  \delta \to \func_{|Q|.|\delta|}$, and will later add the extra component.
  The intuition is as follows.  Consider a path $\und\pi$ on $\und A$ from
  the initial state to a state $\und q$ --- the empty path is considered to
  be from $\und{q_0}$ to $\und{q_0}$.  We view $(V(\und\pi))(\bar{0})$ as a
  list of counters $(\bar{c_1}, \ldots, \bar{c_{|Q|}})$ where $\bar{c_q} \in
  \bbn^{|\delta|}$.  We will ensure that for any $q \in \und q$, $\bar{c_q}$
  is the Parikh image of a path $\pi$ in $A$ from $q_0$ to $q$ such that
  $\mu_A(\pi) = \mu_{\und A}(\und \pi)$.  If two such paths $\pi_1$ and
  $\pi_2$ exist, we may choose one arbitrarily, as they are equivalent in the
  following sense: if $\rho$ is such that $\pi_1\rho \in \apaths(A)$ and
  $\pkh(\pi_1\rho) \in C$, then the same holds for $\pi_2$.

  For $\und p \subseteq Q$, $q \in Q$, and $a \in \Sigma$, let $P(\und p, q,
  a)$ be the smallest $p \in \und p$ such that $S(p, q, a) \neq \bot$ and
  $\bot$ if none exists.  Let $\und t = (\und p, a, \und q)$ be a transition
  of $\und A$.  Define:
  $$V(\und t) = \left(\sum_{q \in \und q} M(P(\und p, q, a),
    q), \qquad \sum_{q \in \und q} N(q, \pkh(S(P(\und p, q,
    a),q,a)))\right)$$%
  where $M(p, q)$ is the matrix which transfers the $p$-th counter to the
  $q$-th, and zeroes the others, and $N(q, \bar{d})$ is the shift of $\bar{d}
  \in \bbn^{|\delta|}$ to the $q$-th counter.  More precisely, $M(p,q)_{i, j}
  = 1$ iff there exists $1 \leq e \leq |\delta|$ such that $i = (q -
  1).|\delta| + e$ and $j = (p - 1).|\delta| + e$; likewise, $N(q,\bar{d}) =
  (0^{(q - 1).|\delta|}) \cdot \bar{d} \cdot (0^{(|Q| - q).|\delta|})$.  The
  matrices appearing in $V$ are 0-1 matrices with at most one nonzero entry
  per row; composing such matrices preserves this property, thus $\mon(V)$ is
  finite.

  \begin{claim*}
    Let $\und\pi$ be a path on $\und A$ from $\und{q_0}$ to some state $\und
    q$.  Let $(\bar{c_1}, \ldots, \bar{c_{|Q|}}) = (V(\und\pi))(\bar{0})$,
    where $\bar{c_q} \in \bbn^{|\delta|}$.  Then for all $q \in \und q$,
    $\bar{c_q}$ is the Parikh image of a path in $A$ from $q_0$ to $q$
    labeled by $\mu(\und\pi)$.
  \end{claim*}

  We show this claim by induction.  If $|\und\pi| = 0$, then the final state
  is $\{q_0\}$ and $\bar{c_{q_0}}$ is by definition all-zero.  This describes
  the empty path in $A$, from $q_0$ to $q_0$.  Let $\und\pi$ be such that
  $|\und\pi| > 0$, and consider a state $q$ in the final state of $\und\pi$.
  Write $\und\pi = \und\rho\cdot\und t$, with $\und t \in \und\delta$, and
  let $p = P(\dest(\und \rho), q, \mu(t))$ and $\zeta = S(p, q, \mu(t))$.
  The induction hypothesis asserts that the $p$-th counter of $(V(\und
  \rho))(\bar{0})$ is the Parikh image of a path $\rho$ on $A$ from $q_0$ to
  $p$ labeled by $\mu(\und\rho)$.  Thus, the $q$-th counter of $(V(\und
  \pi))(\bar{0})$ is $\pkh(\rho) + \pkh(\zeta)$, which is the Parikh image of
  $\rho\zeta$, a path from $q_0$ to $q$ labeled by $\mu(\und\pi)$.  This
  concludes the proof of this claim.

  We now define $U\colon \und\delta \to \func_{|Q|.|\delta| + 1}$.  We add a
  component to the functions of $V$, such that for $\und\pi \in \apaths(\und
  A)$, the last component of $(U(\und \pi))(\bar{0})$ is $0$ if $\dest(\und
  \pi) \cap F = \emptyset$ and $\min(\dest(\und \pi) \cap F)$ otherwise.  For
  $\und t = (\und p, a, \und q) \in \und\delta$, let:
  $$U(\und t): (\bar{x}, s) \mapsto \left((V(\und t))(\bar{x}),\qquad
    \begin{cases}
      q&\text{if $q$ is the smallest s.t. $q \in \und q \cap F$,}\\
      0&\text{if no such $q$ exists.}
    \end{cases}\right)$$%
  Now define $E \subseteq \bbn^{|Q|.|\delta| + 1}$ to be such that
  $(\bar{v_1}, \ldots, \bar{v_{|Q|}}, q) \in E$ iff $\bar{v_q} \in C$.  We
  adjoin $\bar{0}$ to $E$ iff $\bar{0} \in C$, in order to deal with the
  empty word.  A word $w$ is accepted by the DetAPA $(\und A, U, E)$ iff
  there exists a path in $A$ from $q_0$ to $q \in F$, labeled by $w$, and
  whose Parikh image belongs in $C$, i.e., $w \in L(A, C)$.

  Finally, note that $L(\und A) = L(A)$ and that $\mon(U)$ is finite: the
  extra component of $U$ only adds a column and a row of $0$'s to the
  matrices.
\end{proof}

Before concluding with Lemma~\ref{lem:dapafmbd=dpabd}, let us first recall
the following results and definitions:
\begin{lemma}[{\cite[Lemmata 5.5.1 and 5.5.4]{ginsburg66}}]\label{lem:uvz}
  Let $u, v \in \Sigma^*$.  Then $(u + v)^*$ is bounded iff there exists $z
  \in \Sigma^*$ such that $u, v \in z^*$.
\end{lemma}

\begin{definition}[\cite{finkel-iyer-sutre03}]
  Let $\Gamma$ be an alphabet.  A \emph{semilinear regular expression (SLRE)}
  is a finite set of \emph{branches}, defined as expressions of the form
  $y_0x_1^*y_1x_2^*y_2\cdots
  x_n^*y_n$, where $x_i \in \Gamma^+$ and $y_i \in \Gamma^*$.  The language
  of an SLRE is the union of the languages of each of its branches.
\end{definition}

\begin{theorem}[\cite{finkel-iyer-sutre03}]
  A regular language is bounded iff it is expressible as a SLRE.
\end{theorem}

We will need the following technical lemma.  Bounded languages being closed
under morphisms, for all automata $A$ if $\apaths(A)$ is bounded then so is
$L(A)$.  The converse is true when $A$ is deterministic (and false
otherwise):
\begin{lemma}\label{lem:apathbounded} Let $A$ be a deterministic
  automaton for a bounded language, then $\apaths(A)$ is bounded.  Moreover,
  $\apaths(A)$ is expressible as a SLRE whose branches are of the form
  $\rho_0\pi_1^*\rho_1\cdots \pi_n^*\rho_n$ where $\rho_i \neq \eps$ for all
  $1 \leq i < n$ and the first transition of $\pi_i$ differs from that of
  $\rho_i$ for every $i$ (including $i=n$ if $\rho_i \neq \eps$).
\end{lemma}
\begin{proof}
  Recall that bounded languages are closed under \emph{deterministic}
  rational transduction (see, e.g., \cite[Lemma~5.5.3]{ginsburg66}).  Now
  consider $A'$ the automaton defined as a copy of $A$ except that its
  transitions are $(q, (a, t), q')$ for $t = (q, a, q')$ a transition of $A$.
  Then $\tau_{A'}$, the deterministic rational transduction defined by $A'$,
  is such that $\apaths(A) = \tau_{A'}(L(A))$, and thus, $\apaths(A)$ is
  bounded.

  It will be useful to note the claim that if $X\pi_1^*\pi_2^*Y\subseteq
  \apaths(A)$ for some nontrivial paths $\pi_1,\pi_2$ and some bounded
  languages $X$ and $Y$, then for some path $\pi$, $X\pi_1^*\pi_2^*Y\subseteq
  X\pi^*Y \subseteq \apaths(A)$.  To see this, note that if
  $X\pi_1^*\pi_2^*Y\subseteq \apaths(A)$ then $X(\pi_1+\pi_2)^*Y\subseteq
  \apaths(A)$ because $\pi_1$ and $\pi_2$ are loops on a same state.  Now
  $X(\pi_1+\pi_2)^*Y$ is bounded because $\apaths(A)$ is bounded, hence
  $(\pi_1+\pi_2)^*$ is bounded.  So pick $\pi$ such that $\pi_1,\pi_2\in
  \pi^*$ (by Lemma~\ref{lem:uvz}).  Then $X(\pi_1+\pi_2)^*Y\subseteq
  X\pi^*Y$.  But $\pi$ is a loop in $A$ because $\pi_1=\pi^j$ for some $j>0$
  is a loop so that $\orig(\pi)=\dest(\pi)$ in $A$. Hence $X\pi^*Y\subseteq
  \apaths(A)$.  Thus $X\pi_1^*\pi_2^*Y\subseteq X(\pi_1+\pi_2)^*Y\subseteq
  X\pi^*Y\subseteq \apaths(A)$.

  Let $E$ be a SLRE for $\apaths(A)$, and consider one of its branches $P =
  \rho_0\pi_1^*\rho_1\cdots \pi_n^*\rho_n$.  We assume $n$ to be minimal
  among the set of all $n'$ such that $\rho_0\pi_1^*\rho_1\cdots
  \pi_n^*\rho_n\subseteq \rho_0'\pi_1'^*\rho_1'\cdots
  \pi_{n'}'^*\rho_{n'}'\subseteq \apaths(A)$ for some
  $\rho_0',\pi_1',\ldots,\pi_{n'}',\rho_{n'}'$.

  First we do the following for $i=n, n-1, \ldots, 1$ in that order.  If
  $\pi_i=\zeta\pi$ and $\rho_i = \zeta\rho$ for some maximal nontrivial path
  $\zeta$ and for some paths $\pi$ and $\rho$, we rewrite
  $\rho_{i-1}\pi_i^*\rho_i$ as $\rho_{i-1}'\pi_i'^*\rho_i'$ by letting
  $\rho'_{i-1} = (\rho_{i-1}\zeta)$, $\pi'_i = (\pi\zeta)$ and $\rho'_i =
  \rho$. This leaves the language of $P$ unchanged and ensures at the $i$th
  stage that the first transition of $\pi'_j$ (if any) differs from that of
  $\rho'_j$ (if any) for $i\leq j\leq n$.  Note that $n$ has not changed.

  Let $\rho_0'\pi_1'^*\rho_1'\cdots \pi_n'^*\rho_n'$ be the expression for
  $P$ resulting from the above process.  By the minimality of $n$,
  $\pi_i'\neq \eps$ for $1\leq i\leq n$.  And for the same reason,
  $\rho_i'\neq \eps$ for $1\leq i<n$, since $\rho_i'=\eps$ implies
  $X\pi_i'^*\pi_{i+1}'^*Y\subseteq Xz^*Y \subseteq \apaths(A)$ for some $z$
  by the claim above, where $X=\rho_0'\cdots \pi_{i-1}'^*\rho_{i-1}'$ and
  $Y=\rho_{i+1}'\pi_{i+2}'^*\cdots\rho_n'$ are bounded languages.
\end{proof}

We are now ready to prove the following, thus concluding the proof of
Theorem~\ref{thm:main}:
\begin{lemma}\label{lem:dapafmbd=dpabd}
  Let $(A, U, C)$ be a DetAPA such that $L(A)$ is bounded and $\mon(U)$ is
  finite.  Then there exist a finite number of flat DetCA the union of the
  languages of which is $L(A, U, C)$.
\end{lemma}
\begin{proof}
  Let $A = (Q_A, \Sigma, \delta, q_0, F_A)$ be a deterministic automaton
  whose language is bounded, let $U\colon \delta \to \func_d$ for some $d >
  0$ such that $\mon(U)$ is finite, and let $C \subseteq \bbn^d$ be a
  semilinear set.
  
  Consider the set $\apaths(A)$ of accepting paths in $A$; it is, by
  Lemma~\ref{lem:apathbounded}, a bounded language.  Let $P$ be the language
  defined as a branch
  $\rho_0\pi_1^*\rho_1\cdots \pi_n^*\rho_n$ of the SLRE for
  $\apaths(A)$ given by Lemma~\ref{lem:apathbounded}.  We will first
  construct a finite number of flat DetPA for this branch, such that their
  union has language $\{\mu(\pi) \st \pi \in P \land (U(\pi))(\bar{0}) \in
  C\}$.  For simplicity, we will assume that $\rho_0 = \eps$ and $\rho_n \neq
  \eps$.  We will later show how to lift this restriction.

  For $t \in \delta$, we write $U(t) = (M_t, \bar{v_t})$, and more generally,
  for $\pi = t_1\cdots t_k \in \delta^*$, we write $U(\pi) = (M_\pi,
  \bar{v_\pi})$; in particular, $M_\pi = M_{t_k}\ldots M_{t_1}$.

  Let $p_i, r_i \in \bbn^+$, for $1 \leq i \leq n$, be integers such that
  $M_{\pi_i}^{p_i} = M_{\pi_i}^{p_i+r_i}$; such integers are guaranteed to
  exist as $\mon(U)$ is finite.  Now let $a_i$ be in $\{0, \ldots, p_i + r_i
  - 1\}$, for $1 \leq i \leq n$, and note, as usual, $\bar{a} = (a_1, \ldots,
  a_n)$.  For succinctness, we give constructions where the labels of the
  transitions can be nonempty \emph{words}.  This is to be understood as a
  string of transitions, with fresh states in between, with the following
  specificity: for $w = \ell_1\cdots \ell_k$, an extended-transition $(q, (w,
  \bar{d}), q')$ in a PA is a string of transitions the first of which is
  labeled $(\ell_1, \bar{d})$ and the other ones $(\ell_i, \bar{0})$, $i \geq
  2$.

  Let $B_{\bar{a}}$ be the ordinary flat automaton with state set $Q = \{q_i,
  q_i' 
  \st 1 \leq i \leq n\} \cup \{q_{n+1}\}$, and with transition set $\delta_B$
  defined by:
  $$\delta_B = \{(q_i, \pi_i^{a_i}, q_i'), (q_i', \rho_i, q_{i+1}) \st 1 \leq i \leq n\} \cup
  \{(q_i', \pi_i^{r_i}, q_i') \st 1 \leq i \leq n \land a_i \geq p_i\}.$$ We
  let the initial state be $q_1$ and the final state be $q_{n+1}$; note that
  $B_{\bar{a}}$ is deterministic, thanks to the form of $P$ given by
  Lemma~\ref{lem:apathbounded}.

  So $L(B_{\bar{a}}) =
  \{\pi_1^{a_1+m_1.r_1}\rho_1\cdot\hspace{-1pt}\cdot\hspace{-1pt}\cdot
  \pi_n^{a_n+m_n.r_n}\rho_n | (\forall i)[m_i \in \bbn \land (a_i < p_i
  \rightarrow m_i = 0)]\}$, and the union of the languages of the
  $B_{\bar{a}}$'s for all valid values of $\bar{a}$ is $P$.  We now make a
  simple but essential observation on $(U(\pi))(\bar{0})$.  For paths $\pi$
  and $\rho$, we write $V(\pi, \rho)$ for $M_\rho.\bar{v_\pi}$.  Then let
  $\pi \in L(B_{\bar{a}})$, and let $m_i$, for $1 \leq i \leq n$, be such
  that $\pi = \pi_1^{a_1 + m_1.r_1}\rho_1\cdots \pi_n^{a_n+ m_n.r_n}\rho_n$.
  First note that $V(\cdot, \cdots\pi_i^{a_i+m_i.r_i}\cdots) = V(\cdot,
  \cdots\pi_i^{a_i}\cdots)$, then:
  $$\begin{array}{ll}
    (U(\pi))(\bar{0}) & = U(\rho_n) \circ U(\pi_n^{a_n+m_n.r_n}) \circ \cdots
    \circ U(\pi_1^{a_1 + m_1.r_1})(\bar{0}) \\ & =
    M_{\rho_n}(M_{\pi_n}(\ldots(\underbrace{M_{\pi_1}(\ldots(M_{\pi_1}
    }_{a_1 + m_1.r_1 \text{ times}}
    \!\!.\bar{0} +
    \bar{v_{\pi_1}})\ldots) + \bar{v_{\pi_1}})\ldots) +
    \bar{v_{\pi_n}}) + \bar{v_{\rho_n}}\\ & = \sum_{i=1}^n
    \left((\sum_{l=0}^{a_i+m_i.r_i-1} V(\pi_i,
      \pi_i^l\rho_i\pi_{i+1}^{a_{i+1}}\cdots \rho_n)) + V(\rho_i,
      \pi_{i+1}^{a_{i+1}}\cdots \rho_n)\right).
  \end{array}$$ Moreover, letting $\bar{v_i}(l) = V(\pi_i,
  \pi_i^l\rho_i\pi_{i+1}^{a_{i+1}}\cdots\rho_n)$, we have for all $i \in \{1,
  \ldots, n\}$:
  $$
  \begin{array}{ll}
    \sum_{l=0}^{a_i + m_i.r_i - 1} \bar{v_i}(l) & = \sum_{l=0}^{a_i-1}
    \bar{v_i}(l) + \sum_{l=a_i}^{a_i+m_i.r_i-1} \bar{v_i}(l)\\ & =
    \sum_{l=0}^{a_i-1} \bar{v_i}(l) +
    \sum_{m=0}^{m_i-1}\sum_{l=a_i+m.r_i}^{a_i+(m+1).r_i-1} \bar{v_i}(l)\\ & =
    \sum_{l=0}^{a_i-1} \bar{v_i}(l) +
    m_i\times\sum_{l=a_i}^{a_i+r_i-1}\bar{v_i}(l).
  \end{array}$$

  We now define a PA $D_{\bar{a}}$ which is a copy of $B_{\bar{a}}$
  except for the labels of its transition set $\delta_D$.  Each transition of
  $D_{\bar{a}}$ incorporates the relevant value of $V(\cdot, \cdot)$ so that the
  sum in the last equation above will be easily computable.  For each
  transition $(q, \pi, q_i') \in \delta_B$, there is a transition $(q, (\pi,
  \bar{d}), q_i') \in \delta_D$ such that if $q = q_i$ then
  $\bar{d}=\sum_{l=0}^{a_i-1}\bar{v_i}(l)$, and if $q = q_i'$ then $\bar{d} =
  \sum_{l=a_i}^{a_i+r_i-1}\bar{v_i}(l)$.  Finally, for each transition
  $(q_i', \rho_i, q_{i+1}) \in \delta_B$, there is a transition $(q_i',
  (\rho_i, \bar{d}), q_{i+1}) \in \delta_D$ with $\bar{d} = V(\rho_i,
  \pi_{i+1}^{a_{i+1}}\cdots \rho_n)$.

  Let $\omega \in L(D_{\bar{a}})$, there exists $\pi \in L(B_{\bar{a}})$ such
  that $\proj(\omega) = \pi$.  We show by induction on $n$, the number of
  $\pi_i$'s, that $\pext(\omega) = (U(\pi))(\bar{0})$.  If $n = 0$, then it
  is trivially true.  Suppose $n > 0$, and write $\pi =
  \pi_1^{a_1+m.r_1}\rho_1\pi'$ for some $m$ and $\pi'$.  Write $\omega =
  \omega_1\omega_2\omega_3\omega'$ with $\proj(\omega_1) = \pi_1^{a_1}$,
  $\proj(\omega_2)=\pi_1^{m_1.r_1}$, $\proj(\omega_3)=\rho_1$.  Then:
  $$
  \begin{array}{ll}
    (U(\pi))(\bar{0}) & = \underbrace{\textstyle\sum_{l=0}^{a_i-1}
    \bar{v_i}(l)}_{\pext(\omega_1)} +
    \underbrace{m_i\times\textstyle\sum_{l=a_i}^{a_i+r_i-1}\bar{v_i}(l)}_{\pext(\omega_2)}
    + \underbrace{V(\rho_1, \pi')}_{\pext(\omega_3)} +
    \underbrace{(U(\pi'))(\bar{0})}_{\pext(\omega') \text{ by ind. hyp.}}\\ &
    = \pext(\omega).
  \end{array}
  $$

  Thus, for a path $\pi$, $\pi \in L(B_{\bar{a}})$ iff there exists $\omega
  \in L(D_{\bar{a}})$ with $\proj(\omega) = \pi$, and in this case,
  $\pext(\omega) = (U(\pi))(\bar{0})$.

  We now lift the restrictions we set at the beginning of the proof.
  First, if $\rho_0 \neq \eps$, we add a fresh state $q_0'$ to $B_{\bar{a}}$
  together with a transition $(q_0', \rho_0, q_1)$.  In the construction of
  $D_{\bar{a}}$, the label $\rho_0$ is changed to $(\rho_0, V(\rho_0,
  \pi_1^{a_1}\rho_1\cdots \pi_n^{a_n}\rho_n))$.  Now, for $\omega \in
  L(D_{\bar{a}})$, let $\pi = \proj(\omega)$ and write $\pi = \rho_0\pi'$ for
  some $\pi'$; likewise, write $\omega = \omega_1\omega'$ where
  $\proj(\omega_1) = \rho_0$.  Then:
  $$(U(\pi))(\bar{0}) = V(\rho_0, \pi_1^{a_1}\rho_1\cdots \pi_n^{a_n}\rho_n)
  + (U(\pi'))(\bar{0}).$$ The previous results show that $(U(\pi'))(\bar{0})
  = \pext(\omega')$, and thus $(U(\pi))(\bar{0}) = \pext(\omega)$.

  Next, if $\rho_n = \eps$, then the transition $(q_n', \rho_n, q_{n+1})$ is
  removed from the construction of $B_{\bar{a}}$, and the final state changed
  to $q_n$.  The absence of this transition does not influence the
  computations in $D_{\bar{a}}$, thus we still have that for any $\omega \in
  L(D_{\bar{a}})$, $\pext(\omega) = (U(\proj(\omega)))(\bar{0})$.

  Now define $D'_{\bar{a}}$ to be a copy of $D_{\bar{a}}$ which differs only
  in the transition labels: a transition labeled $(t, \bar{d})$ in
  $D_{\bar{a}}$ becomes a transition labeled $(\mu_A(t), \bar{d})$ in
  $D'_{\bar{a}}$.  The properties of the SLRE for $\apaths(A)$ extracted from
  Lemma~\ref{lem:apathbounded}
  imply that $(D'_{\bar{a}}, C)$ is
  a flat DetPA.  Let $L_P \in \LDPA$ be the union of the languages of all
  flat DetPA $(D'_{\bar{a}}, C)$, for any $\bar{a}$ with $a_i \in \{0,
  \ldots, p_i + r_i - 1\}$.  Then $L_P=\{\mu(\pi) \st \pi \in P \land
  (U(\pi))(\bar{0}) \in C\}$.  Thus the union of all $L_P$ for all branches
  $P$ is $L(A, U, C)$.
The flat DetPA involved can be made into flat DetCA because the equivalence
in Theorem~\ref{thm:altdef} carries over the case of flat automata.
\end{proof}

We note that we can show that there exist nonsemilinear bounded languages in
$\LDAPA$, thus there exist bounded languages in $\LDAPA \setminus \LPA$.

\section{Discussion}

We showed that PA and DetPA recognize the same class of bounded languages,
namely BSL.  Moreover, we note that the union of flat DetCA is a concept that
has already been defined in~\cite{bouajjani-habermehl98} as
\emph{1-constrained-queue-content-decision-diagram (1-CQDD)}, and we may thus
conclude, thanks to the specific form of the DetCA obtained in
Theorem~\ref{thm:main}, that 1-CQDD capture exactly BSL.

A related model, \emph{reversal-bounded multicounter machines
  (RBCM)}~\cite{ibarra78}, has been shown to have the same expressive power
as PA~\cite{klaedtke-ruess03}.  We can show that \emph{one-way} deterministic
RBCM are strictly more powerful than DetPA, thus our result carries over to
the determinization of bounded RBCM languages.  This
generalizes~\cite[Theorem~3.5]{ibarra-su99} to bounded languages where the
words of the socle are of length other than one.  The fact that
\emph{two-way} deterministic RBCM are equivalent to \emph{two-way} RBCM over
bounded languages can be found
in~\cite[Theorem~3]{dang-ibarra-bultan-kemmerer-su00} without proof and
in~\cite{lisovik-koval05} with an extremely terse proof.



\paragraph{Acknowledgments.}  We thank the anonymous referees for helpful
comments. The first author would further like to thank
L.~Beaudou, M.~Kaplan and A.~Lema\^itre.

\bibliography{words11}


\end{document}